\documentclass[journal, onecolumn,11pt]{IEEEtran}

\usepackage{amsthm,amssymb,graphicx,amsmath,color,amsfonts,cite,pbox}

\def\nb0{{\mathbf{0}}}
\def\nb1{{\mathbf{1}}}

\newtheorem{lemma}{Lemma}

\newtheorem{cor}{Corollary}

\newtheorem{remark}{Remark}

\def\sir{\mathtt{SIR}}

\title{HetNets and Massive MIMO: Modeling, Potential Gains, and Performance Analysis}

\author{\authorblockN{Marios Kountouris, Nikolaos Pappas}\\
\authorblockA{Department of Telecommunications, SUPELEC (Ecole Sup{\'e}rieure d'Electricit{\'e}),\\ 3 rue Joliot-Curie, 91192 Gif-sur-Yvette, France\\
Email: \{marios.kountouris,nikolaos.pappas\}@supelec.fr
}}

\begin{document}
\maketitle

\begin{abstract}
We consider a heterogeneous cellular network (HetNet) where a macrocell tier with a large antenna array base station (BS) is overlaid with a dense tier of small cells (SCs).
We investigate the potential benefits of incorporating a massive MIMO BS in a TDD-based HetNet and we provide analytical expressions for the coverage probability and the area spectral efficiency using stochastic geometry. The duplexing mode in which SCs should operate during uplink macrocell transmissions is optimized. Furthermore, we consider a reverse TDD scheme, in which the massive MIMO BS can estimate the SC interference covariance matrix. Our results suggest that significant throughput improvement can be achieved by exploiting interference nulling and implicit coordination across the tiers due to flexible and asymmetric TDD operation.
\end{abstract}

\section{INTRODUCTION}
\label{sec:intro}

Massive multiple-input multiple output (MIMO) systems and heterogeneous cellular networks (HetNets) are currently considered as two promising and effective solutions for achieving high data rates and coverage gains and satisfying the galloping demand for ubiquitous broadband mobile services. In HetNets, spatial reuse gains are exploited by an extreme network densification through the deployment of small cells (SCs) \cite{mybook}. Small cell base stations (SBSs) are usually low-power nodes with one or very few antennas and higher throughput and enhanced coverage are achieved by significantly reducing the distance between the small cell users (SUEs) and SBSs. Massive MIMO or large-scale MIMO is a multiuser transmission strategy in which the number of antennas employed at each cell site is increased. Exploiting pilot reuse and channel reciprocity in time division duplex (TDD) networks, a large number of users is served using the additional spatial degrees of freedom \cite{MarzettaTCOM10}. Although they appear to be two opposing paradigms, significant performance gains are expected if these network deployment strategies are carefully designed to coexist. Prior work on massive MIMO showed the high throughput gains under different precoding and receiver strategies \cite{MarzettaTCOM10, JakobRZF, Rusek}. The performance of massive MIMO in random networks with Poisson distributed BS locations has been recently analyzed in \cite{BaiHeaMassMIMO, PrasannaMassMIMO}, however extra care is required when asymptotic results in the number of antennas are combined with stochastic geometry and infinite cellular networks. A two-tier network with a massive MIMO system overlaid with a dense tier of SCs is studied in \cite{JakobICC13}.

In this paper, we consider a TDD-based network architecture where the massive MIMO macrocell tier is overlaid with small cells. We investigate the potential benefits by operating both tiers in a cooperative and synergetic way. In particular, the performance in terms of coverage probability and area spectral efficiency is analytically characterized for both tiers using stochastic geometry.
When the macrocell is in uplink mode, we find the optimal number of SCs that can operate in uplink and downlink modes. In the macrocell downlink mode, channel reciprocity and interference covariance matrix estimation are leveraged to spatially cancel the interference caused to a fraction of randomly selected small cells. Our results reveal an interesting interplay between the number of macro BS antennas, the number of mobile user equiments (MUEs) served and the degrees freedom used for interference cancellation.

\section{SYSTEM MODEL}
\label{sec:sysmodel}
We consider a heterogeneous cellular network where a macro tier is overlaid with small cells. The macro BS (MBS) is equipped with $N$ antennas and serves $K$ single-antenna MUEs. Each SBS employs a single antenna and communicates with a given pre-scheduled single-antenna SUE. The MUEs are assumed to be uniformly distributed in a circle of radius $R_m$ centered at the MBS and independent of the small cell tier placements. MBS and MUEs transmit at constant powers $P_{\rm m}$ and $P_{\rm mu}$, respectively.
We model the locations of SC nodes by a spatial homogeneous Poisson point process (PPP) $\Phi$ of density $\lambda$. A standard pathloss model $\ell(x) = \left\|x\right\|^{-\alpha}$, $\alpha > 2$ is assumed. Channel fading between any pair of antenna is assumed to be Rayleigh with unit mean and independent across nodes. The distance between a typical SBS-SUE pair is constant and is denoted by $d$.
The transmit powers of SCs and SUEs are denoted as $P_{\rm s}$ and $P_{\rm su}$, respectively. The network is interference-limited, i.e. the background noise is not taken into account in the analysis.
Both tiers share the same bandwidth, perfect synchronization is assumed, and transmissions take place over flat fading channels. We use a block fading model in which the channel remains constant during one block and fades independently over each block. TDD-based communication is used here and perfect channel reciprocity is assumed, i.e. the uplink and downlink channels for a given pair are identical. Each MUE sends a pre-assigned orthogonal pilot sequence of length $K$ to the MBS during the training phase. The MBS is assumed to estimate perfectly the reverse link channel to each of its MUEs.
Data transmissions occur in time slots and in each time slot, each SC is configured to be in downlink mode with probability $q$, thus in uplink mode with probability $1-q$. For exposition convenience, $q$ is assumed to be independent across time slots and small cells and identical for all small cells.

\section{UPLINK ANALYSIS}
\label{sec:uplink}
In the uplink data transmission phase, each MUE sends data to the MBS and a maximum ratio combining receiver is employed using the normalized channel estimate of the corresponding MUE. In the small cell tier, SBSs are distributed according to a homogeneous PPP $\Phi_{\rm SDL}$ of density $\lambda q$ whereas the SUE locations form a PPP $\Phi_{\rm SUL}$ of density $\lambda (1-q)$.

\subsection{Macrocell Tier Performance}
We consider first the performance of the macrocell tier. The SIR of MUE $k$ at the MBS is given by
\begin{eqnarray}
\sir_k^{\rm MBS} = \frac{P_{\rm mu}\left\|\mathbf{h}_k\right\|^2 r_k^{-\alpha}}{I_{\rm SUL} + I_{\rm SDL} + I_{\rm MUL}}
\end{eqnarray}
where $\mathbf{h}_k \in \mathbb{C}^{N\times 1}$ is the channel from the $k$-th MUE to the MBS. The interference from uplink SC transmissions, downlink SC transmissions, and uplink data transmission in the macrocell is given, respectively, by
\begin{eqnarray}
I_{\rm SUL} & = & \sum_{i \in \Phi_{\rm SUL}} P_{\rm su}g_i R_i^{-\alpha}, \\
I_{\rm SDL} & = & \sum_{j \in \Phi_{\rm SDL}} P_{\rm s} z_j R_j^{-\alpha}, \\
I_{\rm MUL} & = & \sum_{m=1, m\neq k}^{K} P_{\rm mu}\left|\mathbf{h}_m^H\mathbf{v}_k\right|^2 r_m^{-\alpha},
\end{eqnarray}
where $g_i$ is the channel gain between $i$-th SUE and the MBS, $z_j$ is the channel gain between $j$-th SBS and the MBS, and $\mathbf{v}_k = \mathbf{h}_k/\left\|\mathbf{h}_k\right\|$. The distance between the MBS and the $m$-th MUE is denoted as $r_m$, and $R_i$ ($R_j$) is the distance from $i$-th SUE ($j$-th SBS) to the MBS.

\begin{lemma}
\label{lem:DL_MUE}
The uplink coverage probability of MUE, $\mathbb{P}(\sir_k^{\rm MBS}>T)$, is given by
\begin{eqnarray}
p_c^{\rm UL,MUE} = \int_{0}^{R_m} \sum_{n=0}^{N-1}\frac{(-s)^n}{n!}\frac{d^n}{ds^n}\mathcal{L}_{I}(s)f_r(r)\rm d r,
\end{eqnarray}
with $\mathcal{L}_{I}(s) = \mathcal{L}_{I_{\rm SUL}}(s)\mathcal{L}_{I_{\rm SDL}}(s)\mathcal{L}_{I_{\rm MUL}}(s)$, $s = Tr^{\alpha}/P_{\rm mu}$ and $f_r(r) = \frac{2r}{R_m^2}\mathbb{I}\{0\leq r \leq R_m\}$ where $\mathbb{I}(\cdot)$ is the indicator function. The Laplace transform of the interference terms can be expressed as
\begin{eqnarray}
\mathcal{L}_{I_{\rm SDL}}(s) & = & e^{-\lambda q (P_{\rm s}s)^{\frac{2}{\alpha}}C_{\alpha}}, \\
\mathcal{L}_{I_{\rm SUL}}(s) & = & e^{-\lambda (1-q) (P_{\rm su}s)^{\frac{2}{\alpha}}C_{\alpha}}, \\
\mathcal{L}_{I_{\rm MUL}}(s) & = & \prod_{m=1, m\neq k}^{K}(1+Tr_m^{-\alpha}r^{\alpha})^{-1}.
\end{eqnarray}
where $C_{\alpha} = \frac{2\pi^2}{\alpha}\csc\left(\frac{2\pi}{\alpha}\right)$.
\end{lemma}

\begin{proof}
The coverage probability is $p_c^{\rm UL,MUE} = \mathbb{E}_r\left[\mathbb{P}(\sir_k^{\rm MBS}>T | r)\right]$ and is derived easily using for instance tools from \cite{Andy} noting that $g_i$, $z_j$, and $\left|\mathbf{h}_m^H\mathbf{v}_k\right|^2$ are exponential distributed r.v. and $\left\|\mathbf{h}_k\right\|^2$ is gamma distributed, i.e. Gamma(N,1).
\end{proof}

Using the inequality $x^he^{-\frac{hx}{e}} \leq 1$ for $h \geq 0$, we have that
\begin{cor}
\label{cor:DL_MUE 2}
The uplink coverage probability of MUE is upper bounded by
\begin{eqnarray}
p_c^{\rm UL,MUE} \leq \int_{0}^{R_m} \sum_{n=0}^{N-1}\frac{(s)^n}{n!}\mathcal{L}_{I}(s-n/e)f_r(r)\rm d r.
\end{eqnarray}
\end{cor}

The macrocell area spectral efficiency (ASE) is computed as $\mathcal{T}_{\rm m,UL} = K p_c^{\rm UL,MUE} \log_2(1+T)$.

\begin{remark}
When the channel estimate is not reliable and becomes uncorrelated with the actual channel or when a random unitary $\mathbf{v}_k$ is used for receiving the uplink data, the macrocell ASE is given by \begin{eqnarray}
\mathcal{T}_{\rm m,UL} = K\int_{0}^{R_m} \mathcal{L}_{I}(s)f_r(r)\rm d r.
\end{eqnarray}
In that case, it can be easily shown that $\mathcal{T}_{\rm m,UL}$ is a decreasing function in $q$ for $P_{\rm s} > P_{\rm su}$, i.e. it is optimal to set $q=0$ (all SCs in uplink mode).
\end{remark}


\subsection{Small Cell Tier Performance}
We analyze now the performance in the small cell tier when the macro tier is in uplink. Since the channel gain from a typical SUE to its SBS follows an exponential distribution, similarly to Lemma~\ref{lem:DL_MUE}, we can show that
\begin{lemma}
\label{lem:DL_MUE2}
The coverage probability of a typical SUE is given by
\begin{eqnarray}
\label{eq:UL_SC}
p_c^{\rm UL,SUE} & = & e^{-\lambda q (P_{\rm s}s)^{\frac{2}{\alpha}}C_{\alpha}}e^{-\lambda (1-q) (P_{\rm su}s)^{\frac{2}{\alpha}}C_{\alpha}} \nonumber \\
& \times & \prod_{k=1}^{K}(1+sP_{\rm mu} D_k^{-\alpha})^{-1},
\end{eqnarray}
with $s = T d^{\alpha}/P_{\rm s}$ and $D_k$ denoting the distance between the typical SUE and $k$-th MUE.
\end{lemma}

The uplink coverage probability at the small tier access point $p_c^{\rm UL,SBS}$ is given similarly to (\ref{eq:UL_SC}) with $s = T d^{\alpha}/P_{\rm su}$ and $D_k$ denoting the distance between the typical SC and $k$-th MUE.

The aggregate area spectral efficiency of the small cell tier is given by
\begin{eqnarray*}
\mathcal{T}_{\rm s,UL} = \lambda\left[q p_c^{\rm UL,SUE} + (1-q) p_c^{\rm UL,SBS}\right]\log_2(1+T).
\end{eqnarray*}
The optimal $q$ can be found by taking the derivative of $\mathcal{T}_{\rm s,UL}$ with respect to $q$, however the analytical expression does not provide any useful insights. In Fig.~\ref{fig:Fig1}, we plot the aggregate area spectral efficiency in the small cell tier as a function of the SC density for $K=10$ users, $P_{\rm s} = 100$ mW, $P_{\rm su} = P_{\rm mu} = 5$ mW, $\alpha = 4$, $d = 10$ m, and $T = 5$ dB. First, we observe that the network throughput decreases for $\lambda$ increasing as SCs suffer from increased intra-tier interference. Second, as expected, it is beneficial to deploy more SCs in downlink mode, i.e. $q$ can be higher, in sparse networks (small $\lambda$), whereas small cell uplink mode is preferable in dense networks.

\begin{figure}[t]
\centerline{\includegraphics[width=7cm]{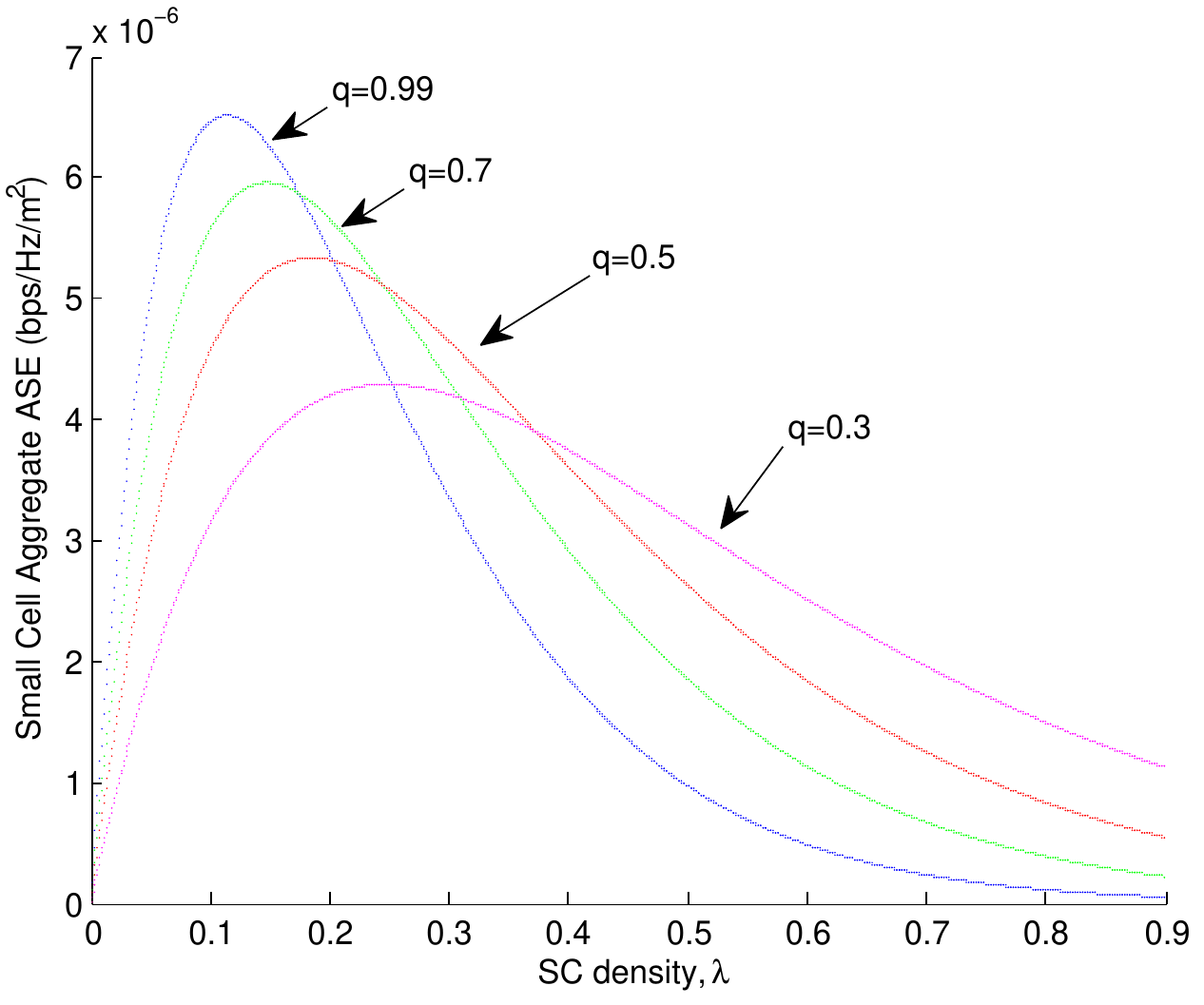}}
\caption{\small Small Cell Cumulative ASE vs. SC density $\lambda$ for different values of $q$.\label{fig:Fig1}}
\end{figure}

\section{DOWNLINK ANALYSIS}
\label{sec:downlink}
In the downlink data transmission phase, the MBS sends data to $K$ MUEs using linear zero-forcing beamforming (LZFB) with equal power allocation across MUEs. We focus on a reverse TDD scheme, i.e. all SCs are in uplink (downlink) when the macrocell operates in downlink (uplink), thus, after the macrocell uplink phase, the MBS can estimate the covariance matrix of the interference from SCs transmissions. Note that it is easier and more reliable to estimate the SC-MBS channels, which are quasi-static for stationary SCs, and that the interference covariance matrix does not vary much over long time scales. Interference covariance knowledge can be exploited and $M \leq N-K$ spatial degrees of freedom can be used to spatially cancel the interference caused by the MBS towards the SCs on the downlink. For exposition convenience, interference cancellation is performed here for $K$ randomly selected SCs. Additional gains are expected using interference nulling to the closest or strongest SCs.

\subsection{Macrocell Tier Performance}
The received SIR at the $k$-th MUE is
\begin{eqnarray}
\sir_k^{\rm MUE} = \frac{\frac{P_{\rm m}}{K}\left|\mathbf{h}_k^H\mathbf{w}_k\right|^2 r_k^{-\alpha}}{I_{\rm SUL}}
\end{eqnarray}
where $\mathbf{w}_k \in \mathbb{C}^{N \times 1}$ is the normalized zero-forcing vectors (columns of the Moore-Penrose pseudoinverse of $\mathbf{H} = \left[\mathbf{h}_1, \ldots, \mathbf{h}_K\right]$).Let $M=\left\lceil \beta(N-K)\right\rceil, 0\leq \beta \leq 1$. The channel gain follows a gamma distribution, i.e. $\left|\mathbf{h}_k^H\mathbf{w}_k\right|^2 \sim \rm Gamma(\theta+1,1)$ with $\theta = \left\lceil (N-K)(1-\beta)\right\rceil$. Using the same procedure as in Lemma~\ref{lem:DL_MUE}, we have that
\begin{lemma}
\label{lem:DL_MUE3}
The coverage probability of $k$-th MUE is given by
\begin{eqnarray}
p_c^{\rm DL,MUE} = \int_{0}^{R_m} \sum_{n=0}^{\theta}\frac{(-s)^n}{n!}\frac{{\rm d}^n}{{\rm d}s^n}\mathcal{L}_{I_{\rm SUL}}(s)f_r(r)\rm d r
\end{eqnarray}
with $s = Tr^{\alpha}K/P_m$ and $\mathcal{L}_{I_{\rm SUL}}(s) = e^{-\lambda s^{2/\alpha}C_{\alpha}}$.
\end{lemma}

\begin{cor}
\label{cor:DL_MUE}
The coverage probability of MUE $k$ for $M = N-K$ ($\beta = 1$) is given by
\begin{eqnarray*}
p_c^{\rm DL,MUE} = \frac{1-e^{-\Delta R_m^2}}{\Delta R_m^2}, \text{with} \ \Delta = \lambda C_{\alpha} \left(\frac{K P_{\rm su}\cdot T}{P_m}\right)^{\frac{2}{\alpha}}.
\end{eqnarray*}
\end{cor}
Expectedly, $p_c^{\rm DL,MUE}$ is a decreasing function with $K$, while the macrocell ASE increases for $K$ increasing.

\subsection{Small Cell Tier Performance}
The SIR at a randomly selected SC depends on whether the interference from the macrocell to this SC is cancelled or not. Let $\mathcal{A}$ be the event that a randomly selected SC is cancelled and $\bar{\mathcal{A}}$ be its complement. Since the average number of SCs in the cell is $\lambda \pi R_m^2$, assuming uniformly distributed SCs in the cell, we have $\mathbb{P}(\mathcal{A}) = \frac{M}{\lambda \pi R_m^2}$.

\begin{lemma}
The coverage probability of a randomly selected SBS at distance $D$ from the MBS is
\begin{eqnarray}
p_c^{\rm DL,SBS} & = & \mathbb{P}(\sir > T | \mathcal{A})\mathbb{P}(\mathcal{A}) \nonumber \\
& + & \mathbb{P}(\sir > T | \bar{\mathcal{A}})(1-\mathbb{P}(\mathcal{A})),
\end{eqnarray}
where $\mathbb{P}(\sir > T | \mathcal{A}) = e^{-\lambda T^{\frac{2}{\alpha}}d^2 C_{\alpha}}$ and
\begin{eqnarray*}
\mathbb{P}(\sir > T | \bar{\mathcal{A}}) = e^{-\lambda T^{\frac{2}{\alpha}}d^2 C_{\alpha}}\left(1+\frac{P_{\rm m}Td^{\alpha}D^{-\alpha}}{K P_{\rm su}}\right)^{-K}.
\end{eqnarray*}
\end{lemma}

The small cell ASE is given by $\mathcal{T}_{\rm DL, SBS} = \lambda p_c^{\rm DL,SBS} \log_2(1+T)$, which is a decreasing in $K$ and increasing in $M$.
In Table~\ref{table1} we compare the macro ASE and the small cell ASE for $\alpha = 4$, $P_{\rm m} = 1$ W, $R_m = 250$m, $\lambda = 10^{-4}$, and $T = 5$ dB. Our results show an interplay between the number of MUEs served ($K$), the number of antennas ($N$), and the number of SCs towards which interference is suppressed ($M$). Doubling the MBS antennas ($N=100 \to 200$) increases the macro throughput by 45\%, while serving more users ($K = 10 \to 20$) results in 33\% ASE increase. More interestingly, increasing $\beta$, i.e. the degrees of freedom used to cancel the interference towards the SCs, from 0.2 to 0.5 reduces the macro ASE by 20\% but can boost the SC ASE by 140\%. Nevertheless, aggressive interference cancelation ($\beta = 1$) significantly reduces the macro ASE (up to 90\%) without providing further SC ASE benefit.

\begin{table}[!t]
\caption {Macro vs. Small Cell ASE Comparison} \label{table1}
\centering
\begin{tabular}{|c|c|c|}
\hline
(K,N,$\beta$) & \pbox{20cm}{Macro ASE \\ (bps/Hz/m$^2$)} & \pbox{20cm}{SC ASE \\ (bps/Hz/m$^2$)}\\
\hline
(10,100, 0.2) & 3.59 & $1.93 \times 10^{-4}$ \\
\hline
(10,100, 0.5) & 2.84 & $4.59 \times 10^{-4}$ \\
\hline
(10,100,1) & 0.37 & $9.2 \times 10^{-4}$ \\
\hline
(10,200, 0.2) & 5.19 & $3.87 \times  10^{-4}$ \\
\hline
(10,200,0.5) & 4.11 & $9.68 \times 10^{-4}$ \\
\hline
(10,200,1) & 0.37 & $1.9 \times 10^{-3}$ \\
\hline
(20,100, 0.2) & 4.79 & $1.63 \times 10^{-4}$ \\
\hline
(20, 100, 0.5) & 3.8 & $4.1 \times 10^{-4}$ \\
\hline
(20,100,1) & 0.53 & $8.15 \times 10^{-4}$ \\
\hline
\end{tabular}
\end{table}

\section{CONCLUSION}
\label{sec:conclusion}
We explored the potential performance gains in a TDD-based HetNet where a massive MIMO BS is overlaid with small cells. We provided analytical expressions to evaluate the coverage probability and area spectral efficiency in both uplink and downlink and we optimized the uplink/downlink mode of small cells. Furthermore, we showed how the macro BS can leverage the reciprocity-based interference covariance matrix knowledge to significantly improve the SC throughput in a reverse-TDD mode.
Future work comprises the performance evaluation with multiple macrocells and different precoders/receivers considering several practical issues, such as spatial correlation between antennas, power control, and pilot contamination. Furthermore, investigating the potential benefits by using massive MIMO systems as a wireless backhaul solution for small cells is a promising research direction.

\section*{Acknowledgement}
{{This research has been partly supported by the ERC Starting Grant 305123 MORE (Advanced Mathematical Tools for Complex Network Engineering). The authors appreciate fruitful discussions with Dr. Jakob Hoydis.}}

\end{document}